\documentclass[11pt,a4paper]{amsart}

\usepackage{amsopn}
\usepackage{amsmath}
\usepackage{amssymb}
\usepackage{amsthm}
\usepackage[foot]{amsaddr}
\usepackage{hyperref}
\usepackage{graphicx}
\usepackage{color}
\usepackage{listings}   
\usepackage{subfig}
\usepackage{enumerate}
\usepackage{todonotes}

\usepackage{qcircuit}
\usepackage{cite}

\newtheorem{definition}{Definition}

\DeclareMathOperator{\Tr}{Tr}

\newcommand{\C}{\ensuremath{\mathbb{C}}}

\newcommand{\ket}[1]{\ensuremath{|#1\rangle}}
\newcommand{\bra}[1]{\ensuremath{\langle#1|}}
\newcommand{\ketbra}[2]{\ensuremath{\ket{#1}\bra{#2}}}
\newcommand{\proj}[1]{\ensuremath{\ketbra{#1}{#1}}}
\newcommand{\braket}[2]{\ensuremath{\langle{#1}|{#2}\rangle}}

\newcommand{\SPAN}{\mathrm{span}}
\newcommand{\Lrm}{\ensuremath{\mathrm{L}}}

\newcommand{\ee}{\ensuremath{\mathrm{e}}}

\newcommand{\ii}{\ensuremath{\mathrm{i}}}

\newcommand{\XX}{\mathcal{X}}
\newcommand{\YY}{\mathcal{Y}}

\newcommand{\PP}{\mathcal{P}}

\renewcommand{\1}{{\rm 1\hspace{-0.9mm}l}}

\newcommand{\ketV}[1]{\ensuremath{|#1\rangle\!\rangle}}
\newcommand{\braV}[1]{\ensuremath{\langle\!\langle#1|}}
\newcommand{\ketbraV}[2]{\ensuremath{\ketV{#1}\braV{#2}}}
\newcommand{\projV}[1]{\ensuremath{\ketbraV{#1}{#1}}}

\newtheorem{lemma}{Lemma}
\newtheorem{theorem}{Theorem}

\newtheorem{corollary}{Corollary}

\title{Discrimination of POVMs with rank-one effects}
\author{Aleksandra Krawiec$^{1,*}$ \and
{\L}ukasz Pawela$^1$ \and 
Zbigniew Pucha{\l}a$^{1,2}$}
\address{$^1$ 
Institute of Theoretical and Applied Informatics, Polish Academy of Sciences, 
ul. Ba{\l}tycka 5, 44-100 Gliwice, Poland}
\address{$^2$ Faculty of Physics, Astronomy and Applied Computer Science, 
Jagiellonian University, ul. {\L}ojasiewicza 11, 30-348 Krak{\'o}w, Poland}
\email{akrawiec@iitis.pl}

\begin{document}
\maketitle

\begin{abstract}
The main goal of this work is to provide an insight into the problem of
discrimination of positive operator valued measures with rank-one effects. It is
our intention to study multiple shot discrimination of such measurements, that
is the case when we are able to use to unknown measurement a given number of
times. Furthermore, we are interested in comparing two possible discrimination
schemes: the parallel and adaptive ones. To this end we construct a pair of
symmetric, information complete positive operator valued measures which can be
perfectly discriminated in a two-shot adaptive scheme. On top of this we provide
an explicit algorithm which allows us to find this adaptive scheme.
\end{abstract}

\section{Introduction}
The transformation of quantum states is at the core of quantum computing. The 
most 
general mathematical tool describing how quantum states transform are quantum 
channels.
While performing computation, one would like to certify that the operation in 
use agrees with the one given by the classical description.
The discrimination task is closely related to the hypothesis testing. 
This is one of the reasons why the problem of discrimination of
quantum channels has attracted a lot of attention
\cite{dariano2001using,gilchrist2005distance,
chiribella2008memory,sacchi2005optimal}. 

Various approaches have been utilized to study the distinguishability of quantum
channels. In work \cite{wang2019resource} the authors introduced the
resource theory of asymmetric distinguishability for quantum channels which was
further developed in \cite{katariya2020evaluating}. A novel scenario called
coherent quantum channel discrimination was studied in \cite{wilde2020coherent}.
Interesting results on the distinguishability concern
specific classes of quantum channels like unitary 
\cite{acin2001statistical,duan2007entanglement}
and Pauli \cite{dariano2005minimax} channels.

We can consider two approaches to channel discrimination -- a parallel and an
adaptive scheme. In the former, we are allowed to perform the unknown channel a
fixed, finite, number of times on an input state, while in the latter case we
can also perform additional operations between the applications of the unknown
channel. In general it is a nontrivial question which discrimination scheme
should be used to distinguish between two quantum channels. It seems natural
that the use of adaptive scheme should improve the discrimination. However, it
may happen that it suffices to use the parallel one. For example in the case
of unitary channels~\cite{acin2001statistical,duan2007entanglement} and von
Neumann measurements it was shown in 
\cite{ji2006identification,puchala2018strategies,puchala2018multiple}
that the parallel scheme is always optimal. It is also known that asymptotically
the use of adaptive strategy does not give an advantage over the parallel one 
for the discrimination of classical \cite{hayashi2009discrimination} and
classical-quantum channels \cite{berta2018amortized}. Nevertheless, an example
of a pair of quantum channels that cannot be distinguished perfectly in
parallel but can be distinguished by the adaptive scheme was proposed in
\cite{harrow2010adaptive}. These channels are however pretty artificial. We will
give another, natural example of two channels which have this property. These
channels are two symmetric informationally complete positive operator valued 
measures (SIC POVMs)  \cite{renes2004symmetric,flammia2006sic} of
dimension three. Moreover, we will show that there exist many pairs of channels
which can be distinguished only adaptively and present numerical results for the
discrimination of random symmetric operator valued measures (POVMs).

In this work we focus on the following scheme. There are two measurements
$\PP_1$ and $\PP_2$. One of them is secretly chosen (with probability
$\frac{1}{2}$) and put into a black box. The black box containing the unknown
measurement device can be used any finite number of times in any configuration.
Two schemes of multiple-shot discrimination will be studied: parallel and
adaptive. The parallel scheme does not require any processing between the uses
of the black box \cite{duan2016parallel}. In the adaptive scheme, however, the
processing between the uses of the black box plays a crucial role. We can
prepare any input state and put it through the discrimination network. Finally,
we perform a known measurement and make a decision which measurement was put
into the black box. 

This work is organized as follows. We begin with preliminaries in 
Section~\ref{sec:preliminaries}. The conditions allowing for perfect 
discrimination in both parallel and adaptive scenarios are stated in 
Section~\ref{sec:parallel_vs_adaptive_discr_scheme}. The example of qutrit SIC 
POVMs which can be discriminated perfectly only by the adaptive scheme, 
as well as the algorithm allowing for perfect discrimination, are 
presented in Section~\ref{sec:sic_povms}. Eventually, in 
Section~\ref{sec:numerical_results} one can find numerical results for the 
discrimination of random rank-one POVMs.

\section{Preliminaries}\label{sec:preliminaries}

Let $\XX$ and $\YY$  be finite-dimensional complex Hilbert spaces.
We denote by $\textrm{L}(\XX,\YY)$ the set of linear operators 
$A:\XX \rightarrow\YY$ and write shortly $\textrm{L}(\XX)$ for 
$\textrm{L}(\XX,\XX)$.
We will use the notation $\mathrm{U}(\XX,\YY)$ for the set of isometry 
operators and $\mathrm{U}(\XX)$ for the set of unitary operators.
A set of quantum states will be denoted $\mathcal{D} (\XX)$.
Quantum operations $\Phi: \textrm{L}(\XX) \rightarrow \textrm{L}(\YY)$ will be 
denoted $\mathrm{T} (\XX,\YY)$ while the subset of quantum channels will be 
denoted $\mathrm{C} (\XX,\YY)$.

We will need the notion of vectorization of a matrix. 
For the identity operator $\1$  of the space $\mathrm{L}(\XX)$ we define its 
vectorization as $\ketV{\1} = \sum_i \ket{ii}$. 
It can be generalized to an arbitrary matrix 
$X \in \mathrm{L}(\XX)$ as $\ketV{X} = (X \otimes \1)\ketV{\1}$.

The most general form of quantum measurements are positive operator valued 
measures (POVMs).
Formally, a collection of positive semidefinite operators   $\PP=\{E_1, 
\ldots,E_m\} \subset \mathrm{L}(\XX)$  is called a POVM iff $\sum_{i=1}^{m} E_i 
= \1$.
The operators $E_i$ are called \emph{effects}.
Every quantum measurement $\PP=\{E_1, \ldots,E_m\}$ can be identified with a 
quantum channel which gives a classical output. Its action on a quantum state 
$\rho$ can be written as
\begin{equation}
\PP (\rho) = \sum_{i=1}^{m} \Tr(E_i \rho) \proj{i}.
\end{equation}

Let us for a moment focus on the distinguishability of general quantum channels.
This problem has been widely studied in the literature in various settings and
utilizing a plethora of mathematical
tools~\cite{chiribella2008quantum,chiribella2009theoretical}.
The probability of correct discrimination between channels $\Phi, \Psi \in
\mathrm{C}(\XX,\YY)$ is upper-bounded by the Holevo-Helstrom theorem
\cite{helstrom1976quantum}
\begin{equation}
p \leq \frac{1}{2} + \frac{1}{4} \|\Phi - \Psi \|_\diamond.
\end{equation}
The diamond norm, known also as the completely bounded trace norm, $\| \cdot 
\|_\diamond $ is defined for a 
Hermiticity-preserving linear map $\Phi \in \mathrm{T}(\XX,\YY)$ as
\begin{equation}
\| \Phi \|_\diamond = \max_{\|X\|_1 = 1}\left\Vert\left(\Phi \otimes 
\1_{\mathrm{L}(\XX)} \right) (X) 
\right\Vert_1.
\end{equation}
where $\| X\|_1$ is defined as a sum of singular values of $X$.
The above-mentioned bound describes the situation when the black box containing 
the unknown POVM can be used only once.
Therefore $\Phi$ and  $\Psi$ are said to be perfectly distinguishable 
in a single-shot scenario iff 
$\|\Phi - \Psi \|_\diamond = 2$.

\subsection{Parallel discrimination}
\label{sec:lower_bound_parallel}
The approach introduced by Holevo-Helstrom theorem can be easily generalized 
into the parallel discrimination scheme. Is such a case the probability of 
correct discrimination after $N$ uses of the black box is upper-bounded by
\begin{equation}
p \leq \frac{1}{2} + \frac{1}{4} \left\Vert\Phi^{\otimes N} - \Psi^{\otimes N} 
\right\Vert_\diamond.
\end{equation}
The parallel discrimination scheme is depicted in Figure 
\ref{fig:parallel_discrimination_scheme}.
Black boxes containing the unknown POVM are denoted by question marks. 
As an input state we use the state on $N+1$ registers denoted by 
$\ket{\psi_{1, \ldots , N+1}}$. After performing the unknown POVM on $N$ 
registers we obtain classical outputs $i_1, \ldots, i_N$. 
Using this classical information we prepare  another known measurement and 
perform it on the last register. Basing on the last measurement's outcome we 
make a decision which measurement was hidden in the black box.

\begin{figure}[h]
\centering\includegraphics{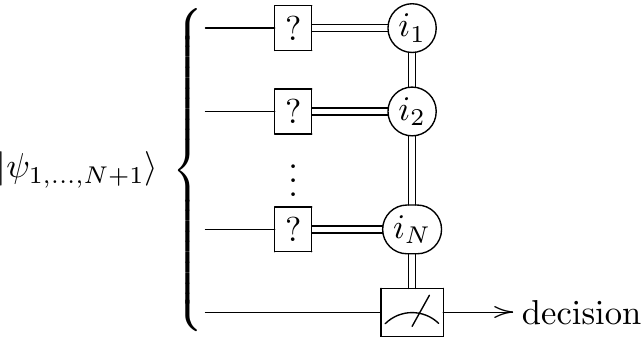}
\caption{Parallel discrimination scheme}
\label{fig:parallel_discrimination_scheme}
\end{figure}

We can state a condition when two measurements can be distinguished perfectly 
in a parallel scenario. It follows directly from the 
Theorem~\ref{th:duan_perfect_distinguishability} which will be presented in the 
subsequent section.
Let $\PP_1=\{\proj{x_1}, \ldots, \proj{x_m}\}$ and 
$\PP_2 =\{\proj{y_1}, \ldots,\proj{y_m}\}$ be POVMs for $m \geq d$.
Then POVMs $\PP_1$ and $\PP_2$ have perfect parallel 
distinguishability if there exists an integer $N>0$ and a quantum state 
$\rho \in \mathcal{D}(\XX^{\otimes N})$ 
such that for all 
multi-indices $i_1, \ldots, i_N$ we have
\begin{equation}
\bra{x_{i_1}  \ldots  x_{i_{N}}} \rho 
\ket{y_{i_1}  \ldots  y_{i_{N}}} = 0. 
\end{equation}

\subsection{Adaptive discrimination scheme}
If we are allowed to perform processing between the uses of black boxes, then 
we arrive at the adaptive discrimination scheme which is depicted in 
Figure \ref{fig:adaptive_discrimination_scheme}.
Although the adaptive discrimination scheme can be used for any number of uses, 
$N$, of the black box, for the sake of clarity let us focus on the case $N=3$. 
We prepare an input state $\ket{\psi}$ and perform the unknown POVM on 
the first register. Basing on the measurement's outcome $i_1$ we perform a 
quantum channel $\Phi_{i_1}$ on the other registers.
Later, we perform the unknown POVM on the second register and obtain the 
classical outcome $i_2$. 
Then we perform a channel $\Phi_{i_1, i_2}$ on the remaining registers which 
now depends on both $i_1$ and $i_2$.
In what follows we perform the unknown POVM on the third register and obtain 
the outcome $i_3$. 
Finally, we perform the known measurement on the last register which now 
depends on all classical outcomes $i_1, i_2, i_3$ and make a decision which of 
the POVMs was hidden in the black box.

\begin{figure}[h]
\centering\includegraphics{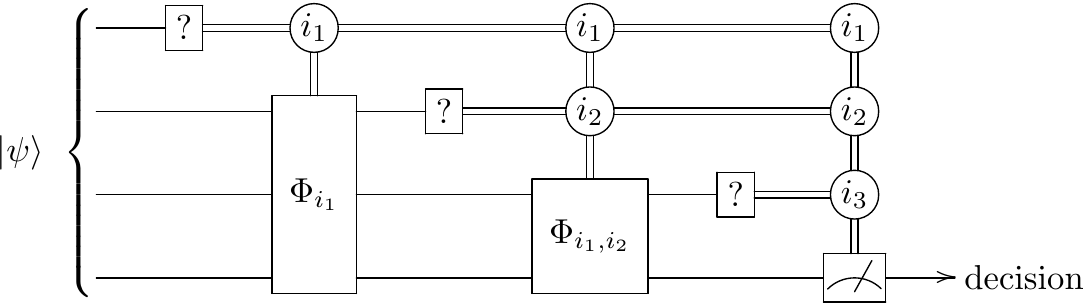}
\caption{Adaptive discrimination scheme}
\label{fig:adaptive_discrimination_scheme}
\end{figure}

\section{Parallel vs adaptive discrimination 
scheme}\label{sec:parallel_vs_adaptive_discr_scheme}
In this section we will state the condition when a pair of quantum operations 
can be distinguished perfectly in the adaptive scheme while the parallel scheme 
does not allow for perfect discrimination. 
Before reviewing some known results on the distinguishability of general 
quantum operations, we need to introduce the notion of disjointness between 
quantum operations.
Let $\rho \in \mathcal{D}(\XX)$ be a quantum state with spectral decomposition 
$\rho = \sum_i 
\lambda_i \proj{\lambda_i}$. Then we define its support as
$\textrm{supp} (\rho) = \textrm{span} \{ \ket{\lambda_i}: \lambda_i > 0 \}$.
\begin{definition}
Two quantum states $\rho_1, \rho_2 \in \mathcal{D}(\XX)$ are called disjoint 
when 
\begin{equation}
\textrm{supp} (\rho_1)  \cap \textrm{supp} (\rho_2) = \{0\}.
\end{equation}
\end{definition}
The notion of disjointness can be generalized to quantum operations. Thus, 
following \cite{duan2009perfect}, we 
have the following definition.
\begin{definition}
Two quantum operations $\Phi, \Psi \in \mathrm{T}(\XX,\YY)$ are called 
(entanglement-assisted) disjoint if there exists an input state $\ket{\psi} \in 
\XX \otimes \mathcal{Z}$ such that 
$\left(\Phi \otimes \1_{\mathrm{L}(\mathcal{Z})}\right)(\proj{\psi})$ and  
$\left(\Psi \otimes \1_{\mathrm{L}{(\mathcal{Z})}}\right)(\proj{\psi})$
are disjoint.
\end{definition}
Let  $\Phi, \Psi \in \mathrm{T}(\XX,\YY)$ be quantum operations  
with Kraus operators $\{E_i\}_i$ and $\{F_i\}_i$ respectively.
From \cite{duan2009perfect} we have a condition for perfect 
distinguishability by a finite number of uses given by the following theorem.

\begin{theorem}\label{th:duan_perfect_distinguishability}
Quantum operations 
$\Phi:\rho \mapsto \sum_i E_i \rho E_i^\dagger $
and $\Psi: \rho \mapsto \sum_j F_j \rho F_j^\dagger$
are perfectly 
distinguishable by a finite 
number of uses iff they are disjoint and 
$\1 \not\in \mathrm{span}\{E_i^\dagger F_j\}_{i,j}$. 
\end{theorem}

The following corollary comes from \cite{duan2016parallel} and gives a 
condition when quantum operations cannot be distinguished in parallel.

\begin{corollary}\label{cor:Duan_parallel_distinguishability}
Using notation as in Theorem~\ref{th:duan_perfect_distinguishability}, if 
$\mathrm{span} \{E_i^\dagger F_j\}_{i,j}$ contains a positive operator 
$\rho >0 $, then 
$\Phi$ and $\Psi$ cannot be distinguished perfectly in parallel by any finite 
number of uses.
\end{corollary}

Combining the above we obtain the corollary which states a condition when 
operations can be distinguished only by the adaptive scheme, that is when 
parallel discrimination is not sufficient but adaptive discrimination allows 
for perfect discrimination.  
\begin{corollary}
Using notation as in Theorem~\ref{th:duan_perfect_distinguishability}, if 
$\mathrm{span} \{E_i^\dagger F_j\}_{i,j}$ contains a positive 
operator $\rho > 0$,  operations $\Phi$ and $\Psi$ are (entanglement-assisted) 
disjoint and 
$\1 \not\in \mathrm{span}\{E_i^\dagger F_j\}_{i,j}$, then 
$\Phi$ and $\Psi$ can be distinguished perfectly only by the adaptive scheme.
\end{corollary}
Let $\PP_1$ and $\PP_2$ be rank-one POVMs with effects 
$\{\proj{x_i}\}_i$ and $\{\proj{y_i}\}_i$ respectively.
Therefore the Kraus operators of $\PP_1$ and $\PP_2$ can be chosen as
$E_i = \{ \ketbra{i}{x_i}\}_i$ and $F_i = \{ \ketbra{i}{y_i}\}_i$ respectively.
Hence in order to check whether 
$\1 \notin  \textrm{span} \{E_i^\dagger F_j\}_{i,j}$ 
it suffices to see whether 
$\1 \notin  \textrm{span} \{ \ketbra{x_i}{y_i}\}_i$ as
$E_i^\dagger F_j = \ketbra{x_i}{y_j} \delta_{ij}$.

\section{SIC POVMs}\label{sec:sic_povms}
In this section we give an example of a pair of quantum measurements 
for which the use of parallel discrimination scheme does not 
give perfect distinguishability but they can be distinguished perfectly by the 
adaptive scheme.
Let us begin with a definition.
A POVM  $\PP=\{E_1, \ldots, E_{d^2}  \} \subset \mathrm{L}(\XX_d)$ of dimension 
$d$ is called a 
SIC POVM  if $E_i = \frac{1}{d} \proj{\phi_i}$ and 
\begin{equation}\label{eq:sic_property}
\left|  \braket{\phi_i}{\phi_j}    \right|^2 = \frac{1}{d+1}  
\end{equation}
for every $i,j = 1, \ldots, d^2$, $i\neq j$.
It is worth mentioning here that it is an open question whether SIC POVMs exist 
in all dimensions. So far numerical results are known for all dimensions up to 
$151$ and for a few other dimensions up to $844$ \cite{fuchs2017sic}.

In what follows we will construct an example of a pair of quantum channels 
which cannot be distinguished perfectly in the parallel scheme but can be 
distinguished perfectly in the adaptive scheme.
Let $\PP_1$ and $\PP_2$ be qutrit SIC POVMs with effects 
$\{\proj{x_i}\}_i$, where $\ket{x_i}=\frac{1}{\sqrt{3}}\ket{\phi_i}$, and 
$\{\proj{y_i}\}_i$ respectively.
Let $\pi = (9, 8, 7, 3, 1, 2, 6, 4, 5 )$ be a permutation such that $\ket{y_i} 
= \ket{x_{\pi(i)}}$ for 
$ i=1, \ldots, 9$.
A detailed construction of the  POVM $\{\proj{x_i}\}_i$ can be found in 
\cite{grassl2017fibonacci} and the exact form of the fiducial vector is stated 
in the supplementary material of \cite{scott2010symmetric}.
In this case $\1 \not\in \mathrm{span}\{\ketbra{x_i}{y_i}\}_{i}$  and 
$\PP_1, \PP_2 \in \textrm{C}(\XX,\YY)$ are (entanglement-assisted) disjoint. 
To see the latter define two states
\begin{equation}
\begin{split}
\rho_1 = (\PP_1 \otimes \1 ) \left(\frac{1}{d}\projV{\1}\right)  = 
\frac{1}{d}\sum_i 
\proj{i} \otimes \proj{x_i}; \\
\rho_2 = (\PP_2 \otimes \1 ) \left(\frac{1}{d}\projV{\1}\right) = 
\frac{1}{d}\sum_i 
\proj{i} \otimes \proj{y_i}. 
\end{split}
\end{equation}
We need to check whether 
$\textrm{supp}(\rho_1) \cap \textrm{supp}(\rho_2) = \{0\}$. We have
\begin{equation}
\begin{split}
\textmd{supp}(\rho_1) = \textrm{span} \{\ket{i} \otimes \ket{x_i}\}_i; \\
\textmd{supp}(\rho_2) = \textrm{span} \{\ket{i} \otimes \ket{y_i}\}_i.
\end{split}
\end{equation}
Take $\ket{\varphi} \in \textrm{supp}(\rho_1) \cap \textrm{supp}(\rho_2)$, then 
there exist numbers $\alpha_1, \ldots \alpha_m$ and $\beta_1, \ldots \beta_m$ 
such that
\begin{equation}
\ket{\varphi} = \sum_i \alpha_i \ket{i} \otimes \ket{x_i} = 
\sum_i \beta_i \ket{i} \otimes \ket{y_i}.
\end{equation}
As $\ket{x_i}$ and $\ket{y_i}$ are 
linearly independent for all $i = 1, \ldots, k$, then $\ket{\varphi} = 0$.

Therefore from Theorem \ref{th:duan_perfect_distinguishability} we obtain that
$\PP_1$ and $\PP_2$ are perfectly distinguishable by a finite number of uses.
Non-existence of the parallel scheme follows from Corollary
\ref{cor:Duan_parallel_distinguishability}, that is  $\PP_1$ and $\PP_2$ cannot
be distinguished in the parallel scheme as there exists a quantum state $\rho
\in \mathcal{D}(\XX)$ such that $\rho \in
\mathrm{span}\{\ketbra{x_i}{y_i}\}_{i}$. The algorithm finding the state $\rho$
will be presented in Section ~\ref{sec:numerical_results}.

\begin{figure}[h]
\centering\includegraphics{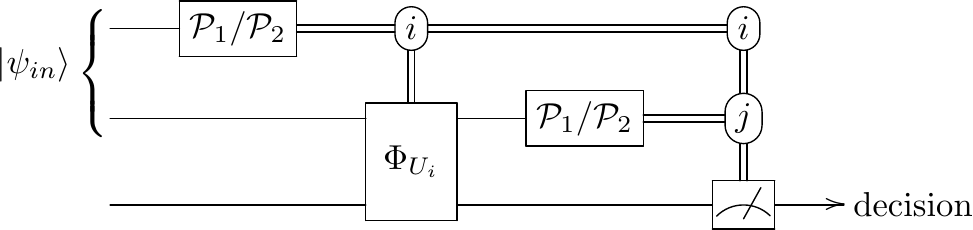}
\caption{Adaptive discrimination scheme}
\label{fig:adaptive_scheme_two_uses}
\end{figure}

Now we focus on the adaptive scheme for perfect 
discrimination between $\PP_1$ and $\PP_2$ which is depicted in Figure 
\ref{fig:adaptive_scheme_two_uses}.
The top register corresponds to the space $\XX$ while the second and third 
registers correspond to spaces $\YY$ and $\mathcal{Z}$ respectively where each 
of these spaces is $\C^3$. 
Our goal is as follows.
Before the final use of the black box we would like to  obtain states 
$\ket{\xi_i}$ and $\ket{\eta_i}$ for which 
$ \left(\PP_1 \otimes \1\right) \left( \proj{\xi_i}\right)$
and
$ \left(\PP_2 \otimes \1\right) \left( \proj{\eta_i}\right)$
are orthogonal.
Therefore we need to find processing $\Phi_i$ for which
\begin{equation}
\begin{split}
\sum_i \proj{i} \otimes \proj{\xi_i}
&= \sum_i \proj{i} \otimes \Phi_i 
\big(\left( \PP_1 \otimes \1 
\right)(\proj{\psi_{in}}) 
\big) \\
\sum_i \proj{i} \otimes \proj{\eta_i}
&= \sum_i \proj{i} \otimes \Phi_i 
\big(\left( \PP_2 \otimes \1 
\right)(\proj{\psi_{in}}) 
\big)
\end{split}
\end{equation}
where $\ket{\psi_{in}}$ is the input state for the discrimination procedure.

The scheme for discrimination between  $\PP_1$ and $\PP_2$ is as follows.
First, using the fact that 
$\textrm{rank} \left(\mathrm{span}\{\ketbra{x_i}{y_i}\}_{i}\right) < 9$,  
we can find a nonzero matrix $A \in \textrm{L}(\XX)$ such that 
\begin{equation*}
A \perp \mathrm{span}\{\ketbra{x_i}{y_i}\}_{i}.
\end{equation*}
From the singular value decomposition we can write
$A = U\Sigma V^\dagger$. 
Next, find vectors 
$\tilde{\ket{\xi}}, \tilde{\ket{\eta}} \in \YY \otimes \mathcal{Z}$ as 
$\tilde{\ket{\xi}} = \ketV{U \sqrt{\Sigma}}$ and
$\tilde{\ket{\eta}} = \ketV{V \sqrt{\Sigma}}$.
We can see that they fulfill the condition
\begin{equation*}
\Tr_\mathcal{Z} \tilde{\ket{\xi}} \tilde{\bra{\eta}} = A.
\end{equation*}
Let $\ket{\xi} = \frac{\tilde{\ket{\xi}}}{\left\Vert \tilde{\ket{\xi}} 
\right\Vert}$ and
$\ket{\eta} = \frac{\tilde{\ket{\eta}}}{\left\Vert \tilde{\ket{\eta}} 
\right\Vert}$.

Now we present the algorithm for the adaptive discrimination between $\PP_1$ 
and $\PP_2$.
\begin{enumerate}
\item As the input state to the discrimination we use the maximally entangled 
state on
first two registers, that is $\ket{\psi_{in}} =
\frac{1}{\sqrt{d}}\ketV{\1_\XX}$. Then we perform the unknown measurement on the
first register, either $\PP_1$ or $\PP_2$, and obtain a label $i$. At the same
time on the second register we obtain the state either $\ket{\psi_i}$ or
$\ket{\varphi_i} \in \YY$. Using the property in Eq. \eqref{eq:sic_property} 
one can 
note that  
\begin{equation*}
\left| \braket{\psi_i}{\varphi_i} \right|  = \left| \Tr (A) \right|^2 = 
\frac{1}{2}.
\end{equation*}

\item Basing on the label $i$ we perform an isometric channel 
$\Phi_{U_i}(\cdot) = U_i\cdot U_i^\dagger$ for  
$ U_i \in \textrm{U}(\YY, \YY \otimes \mathcal{Z})$
such that
\begin{equation*}
\begin{split}
U_i \ket{\psi_i} &= \ee^{\ii \theta_i}\ket{\xi} \\
U_i \ket{\varphi_i} &= \ket{\eta}
\end{split}
\end{equation*}
where $\theta_i = \beta - \alpha_i$ for 
$\braket{\psi_i}{\varphi_i} = r_i\ee^{\ii \alpha_i} $ and
$\braket{\xi}{\eta} = r\ee^{\ii \beta}$.

\item After performing the isometric channel we have a new third register 
$\mathcal{Z}$. 
Then we perform the unknown measurement on the second 
register $\YY$ and obtain a label $j$. 

\item Finally, we measure the third register by some known measurement and make 
a decision whether in the black box there was either $\PP_1$ or $\PP_2$.
\end{enumerate}

\section{Numerical results}\label{sec:numerical_results}
In this section we study numerically the distinguishability of 
random POVMs with rank-one effects.
We check the conditions for perfect distinguishability.
To choose at random a POVM of dimension $d$ with $m$ rank-one
effects we take a Haar random isometry matrix of dimension $d \times m$. 
Then we take projectors onto the columns $\{\ket{x_i}\}_{i=1}^m$ of this matrix 
and obtain a POVM with effects $\{\proj{x_i}\}_{i=1}^m$.

Let $\PP_1$ and $\PP_2$ be POVMs with effects
$\{\proj{x_i}\}_{i=1}^m$ and $\{\proj{y_i}\}_{i=1}^m$  respectively.
If $m=d^2$, then $\1 \in \mathrm{span} \{ \ketbra{x_i}{y_i} \}$  
and therefore any two random rank-one POVMs cannot be distinguished perfectly 
by any finite number of uses.
If $d \leq m < d^2$, then $\1 \not\in \mathrm{span} \{ \ketbra{x_i}{y_i} \}$.
Hence there exists a finite number $N$ of uses which allows for 
perfect distinguishability iff $\PP_1$ and $\PP_2$ are (entanglement-assisted) 
disjoint.
Fortunately, we have the following lemma.

\begin{lemma}
Consider two independently chosen at random POVMs with rank-one effects.
Then with probability one they are (entanglement-assisted) disjoint.
\end{lemma}

\begin{proof}
Define two states
\begin{equation}
\begin{split}
\rho_1 = (\PP_1 \otimes \1 ) \left(\frac{1}{d}\projV{\1}\right)  = 
\frac{1}{d}\sum_i 
\proj{i} \otimes \proj{x_i}; \\
\rho_2 = (\PP_2 \otimes \1 ) \left(\frac{1}{d}\projV{\1}\right) = 
\frac{1}{d}\sum_i 
\proj{i} \otimes \proj{y_i}. 
\end{split}
\end{equation}
We need to check whether 
$\textrm{supp}(\rho_1) \cap \textrm{supp}(\rho_2) = \{0\}$. We have
\begin{equation}
\begin{split}
\textmd{supp}(\rho_1) = \textrm{span} \{\ket{i} \otimes \ket{x_i}\}_i; \\
\textmd{supp}(\rho_2) = \textrm{span} \{\ket{i} \otimes \ket{y_i}\}_i.
\end{split}
\end{equation}
Take $\ket{\varphi} \in \textrm{supp}(\rho_1) \cap \textrm{supp}(\rho_2)$, then 
there exist $\alpha_1, \ldots \alpha_m$ and $\beta_1, \ldots \beta_m$ such that
\begin{equation}
\ket{\varphi} = \sum_i \alpha_i \ket{i} \otimes \ket{x_i} = 
\sum_i \beta_i \ket{i} \otimes \ket{y_i}.
\end{equation}
As it holds that for random POVMs $\ket{x_i}$ and $\ket{y_i}$ are 
linearly independent with probability one for all $i = 1, \ldots, k$, then 
$\ket{\varphi} = 0$.
\end{proof}
arxiv
Knowing that for $m< d^2$ two POVMs with rank-one effects can always be 
distinguished perfectly by a 
finite number of uses, now we are interested when it suffices to use the 
parallel scheme, instead of the adaptive one.
From Theorem \ref{cor:Duan_parallel_distinguishability} we know that if there 
exists a positive operator $\rho \in \mathrm{span} (\ketbra{x_i}{y_i})$, then 
we need to use the adaptive scheme. 
Simple numerical calculations show that there exist pairs of qubit POVMs with 
$3$ rank-one effects which can be distinguished perfectly only by the adaptive 
scheme. 
Moreover, the probability of finding a pair of POVMs with this property is 
roughly equal to $\frac{4}{10}$.

Now we proceed to studying the distinguishability of randomly chosen rank-one 
POVMs in higher dimensions depending on the number of effects.
We study this problem numerically for the input dimension $d=7$ and the number
of effects $d < m < d^2$. The procedure is as follows. We sample $N=10^6$ pairs
of random POVMs in a manner explained at the beginning of this section. For
further details on sampling we refer the reader
to~\cite{Gawron2018,heinosaari2019random}. 

\begin{figure}[h]
\centering\includegraphics{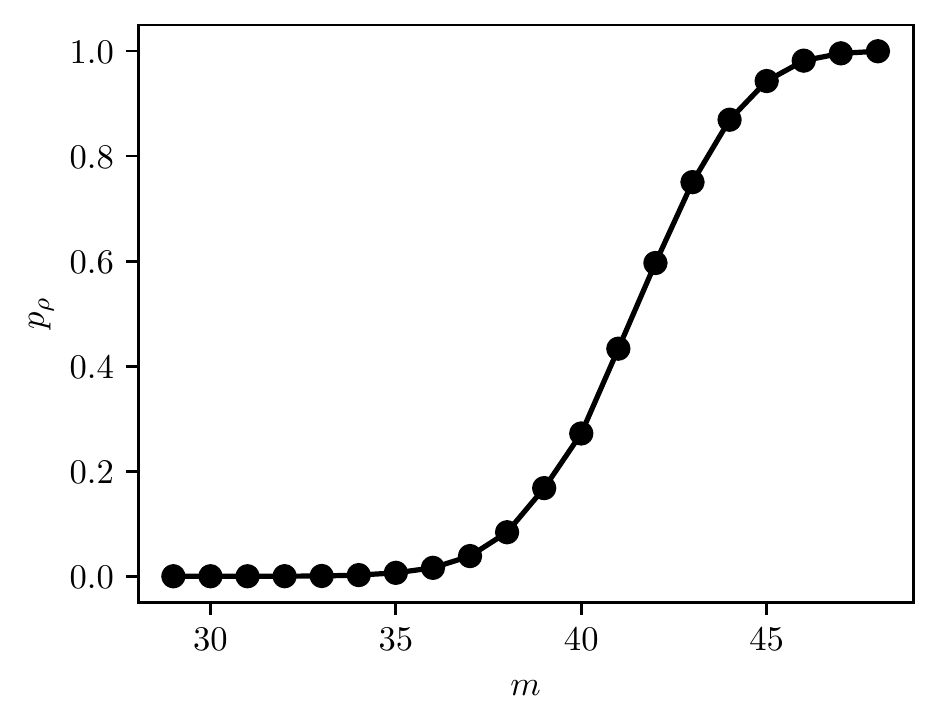}
\caption{
Numerically estimated lower bound for the probability $p_\rho$ of an event that 
for two randomly sampled POVMs with rank-one effects the adaptive scenario is 
necessary for achieving perfect discrimination.
The probability is plotted as a function  
of the number of effects $m$ for the dimension $d=7$. The number of samples per 
point is $N=10^6$.}\label{fig:prob-parallel}
\end{figure}

In order to check whether the
condition stated in Corollary~\ref{cor:Duan_parallel_distinguishability} is
fulfilled, we introduce an iterative algorithm. First, we construct a projection
operator $P$ on the space $\SPAN\{\ketbra{x_i}{y_i}\}_{i=1}^m$ and choose an
initial operator $X \in \Lrm(\XX_d)$. We project $X$ onto the subspace given by
$P$ obtaining some operator $Y$. We substitute the operator $X$ with the density
operator closest to $Y$ and repeat this procedure until it converges or a
predefined number of steps is reached. The final step is to check whether the
resulting operator is of full rank.

Figure~\ref{fig:prob-parallel} shows numerically found probability of finding a
positive operator in the subspace defined in
Corollary~\ref{cor:Duan_parallel_distinguishability} as a function of the number
of effects of POVMs. As we can see in the figure, as the number of effects
increases, so does this probability. This is in agreement with previous results.
On the one hand for the case when $m=d$, we have a von Neumann POVM and we have
shown~\cite{puchala2018multiple} that such measurements can always be
distinguished in a finite number of steps in parallel. On the other hand, when
we have $m=d^2$, then $\SPAN{ \{\ketbra{x_i}{y_i}\}_i}$ gives the entire space
$\C^d$ and we are guaranteed to find the identity operator in it, which implies
that perfect distinguishability cannot be obtained for any finite number of
uses. In general, we are more likely to need to utilize the adaptive scheme for
perfect discrimination when the number of effects of POVMs is close to $d^2$.

\section{Conclusions}
In this work we studied the problem of discriminating POVMs with rank-one
effects. This class of measurements requires a more careful
approach compared to von Neumann POVMs or even general projective measurements.
We show that for this class of measurements there exist instances where the
adaptive discrimination scheme is optimal. This is contrary to von Neumann
POVMs, in which the parallel scheme is always optimal. We illustrate this result
using a special class of POVMs -- the symmetric information complete POVMs. For 
a given input dimension $d$ we study the discrimination between a SIC POVM and 
the same SIC POVM, but with permuted effects. This allows us to construct a 
simple example of two quantum channels which can be perfectly discriminated in 
a finite number of steps but require an adaptive scheme. We state  
the construction of such a discrimination scheme explicitly. Finally, we 
provide some numerical insight on the probability of randomly sampling a pair 
of POVMs with rank-one effects which can be distinguished only in the adaptive 
scenario.

\section*{Acknowledgments}
AK and ZP acknowledge financial support by the Foundation for Polish Science 
through TEAM-NET project (contract no. POIR.04.04.00-00-17C1/18-00). 
{\L}P acknowledges the support of the Polish National Science Centre under the 
project number 2016/22/E/ST6/00062.

\bibliographystyle{ieeetr}
\bibliography{sic_povm}

\begin{thebibliography}{10}

\bibitem{dariano2001using}
G.~M. D'Ariano, P.~L. Presti, and M.~G. Paris, ``Using entanglement improves
  the precision of quantum measurements,'' {\em Physical Review Letters},
  vol.~87, no.~27, p.~270404, 2001.

\bibitem{gilchrist2005distance}
A.~Gilchrist, N.~K. Langford, and M.~A. Nielsen, ``Distance measures to compare
  real and ideal quantum processes,'' {\em Physical Review A}, vol.~71, no.~6,
  p.~062310, 2005.

\bibitem{chiribella2008memory}
G.~Chiribella, G.~M. D’Ariano, and P.~Perinotti, ``Memory effects in quantum
  channel discrimination,'' {\em Physical Review Letters}, vol.~101, no.~18,
  p.~180501, 2008.

\bibitem{sacchi2005optimal}
M.~F. Sacchi, ``Optimal discrimination of quantum operations,'' {\em Physical
  Review A}, vol.~71, no.~6, p.~062340, 2005.

\bibitem{wang2019resource}
X.~Wang and M.~M. Wilde, ``Resource theory of asymmetric distinguishability for
  quantum channels,'' {\em Physical Review Research}, vol.~1, no.~3, p.~033169,
  2019.

\bibitem{katariya2020evaluating}
V.~Katariya and M.~M. Wilde, ``Evaluating the advantage of adaptive strategies
  for quantum channel distinguishability,'' {\em arXiv preprint
  arXiv:2001.05376}, 2020.

\bibitem{wilde2020coherent}
M.~M. Wilde, ``Coherent {Q}uantum {C}hannel {D}iscrimination,'' {\em arXiv
  preprint arXiv:2001.02668}, 2020.

\bibitem{acin2001statistical}
A.~Acin, ``Statistical distinguishability between unitary operations,'' {\em
  Physical Review Letters}, vol.~87, no.~17, p.~177901, 2001.

\bibitem{duan2007entanglement}
R.~Duan, Y.~Feng, and M.~Ying, ``Entanglement is not necessary for perfect
  discrimination between unitary operations,'' {\em Physical Review Letters},
  vol.~98, no.~10, p.~100503, 2007.

\bibitem{dariano2005minimax}
G.~M. D’Ariano, M.~F. Sacchi, and J.~Kahn, ``Minimax discrimination of two
  {P}auli channels,'' {\em Physical Review A}, vol.~72, no.~5, p.~052302, 2005.

\bibitem{ji2006identification}
Z.~Ji, Y.~Feng, R.~Duan, and M.~Ying, ``Identification and distance measures of
  measurement apparatus,'' {\em Physical Review Letters}, vol.~96, no.~20,
  p.~200401, 2006.

\bibitem{puchala2018strategies}
Z.~Pucha{\l}a, {\L}.~Pawela, A.~Krawiec, and R.~Kukulski, ``Strategies for
  optimal single-shot discrimination of quantum measurements,'' {\em Physical
  Review A}, vol.~98, no.~4, p.~042103, 2018.

\bibitem{puchala2018multiple}
Z.~Pucha{\l}a, {\L}.~Pawela, A.~Krawiec, R.~Kukulski, and M.~Oszmaniec,
  ``Multiple-shot and unambiguous discrimination of von {N}eumann
  measurements,'' {\em arXiv preprint arXiv:1810.05122}, 2018.

\bibitem{hayashi2009discrimination}
M.~Hayashi, ``Discrimination of two channels by adaptive methods and its
  application to quantum system,'' {\em IEEE Transactions on Information
  Theory}, vol.~55, no.~8, pp.~3807--3820, 2009.

\bibitem{berta2018amortized}
M.~Berta, C.~Hirche, E.~Kaur, and M.~M. Wilde, ``Amortized channel divergence
  for asymptotic quantum channel discrimination,'' {\em arXiv preprint
  arXiv:1808.01498}, 2018.

\bibitem{harrow2010adaptive}
A.~W. Harrow, A.~Hassidim, D.~W. Leung, and J.~Watrous, ``Adaptive versus
  nonadaptive strategies for quantum channel discrimination,'' {\em Physical
  Review A}, vol.~81, no.~3, p.~032339, 2010.

\bibitem{renes2004symmetric}
J.~M. Renes, R.~Blume-Kohout, A.~J. Scott, and C.~M. Caves, ``Symmetric
  informationally complete quantum measurements,'' {\em Journal of Mathematical
  Physics}, vol.~45, no.~6, pp.~2171--2180, 2004.

\bibitem{flammia2006sic}
S.~T. Flammia, ``On {SIC}-{POVM}s in prime dimensions,'' {\em Journal of
  Physics A: Mathematical and General}, vol.~39, no.~43, p.~13483, 2006.

\bibitem{duan2016parallel}
R.~Duan, C.~Guo, C.-K. Li, and Y.~Li, ``Parallel distinguishability of quantum
  operations,'' in {\em 2016 IEEE International Symposium on Information Theory
  (ISIT)}, pp.~2259--2263, IEEE, 2016.

\bibitem{chiribella2008quantum}
G.~Chiribella, G.~M. D’Ariano, and P.~Perinotti, ``Quantum circuit
  architecture,'' {\em Physical Review Letters}, vol.~101, no.~6, p.~060401,
  2008.

\bibitem{chiribella2009theoretical}
G.~Chiribella, G.~M. D’Ariano, and P.~Perinotti, ``Theoretical framework for
  quantum networks,'' {\em Physical Review A}, vol.~80, no.~2, p.~022339, 2009.

\bibitem{helstrom1976quantum}
C.~W. Helstrom, {\em Quantum {D}etection and {E}stimation {T}heory}.
\newblock Elsevier, 1976.

\bibitem{duan2009perfect}
R.~Duan, Y.~Feng, and M.~Ying, ``Perfect distinguishability of quantum
  operations,'' {\em Physical Review Letters}, vol.~103, no.~21, p.~210501,
  2009.

\bibitem{fuchs2017sic}
C.~A. Fuchs, M.~C. Hoang, and B.~C. Stacey, ``The {SIC} question: History and
  state of play,'' {\em Axioms}, vol.~6, no.~3, p.~21, 2017.

\bibitem{grassl2017fibonacci}
M.~Grassl and A.~J. Scott, ``Fibonacci-{L}ucas {SIC}-{POVM}s,'' {\em Journal of
  Mathematical Physics}, vol.~58, no.~12, p.~122201, 2017.

\bibitem{scott2010symmetric}
A.~J. Scott and M.~Grassl, ``Symmetric informationally complete
  positive-operator-valued measures: A new computer study,'' {\em Journal of
  Mathematical Physics}, vol.~51, no.~4, p.~042203, 2010.

\bibitem{Gawron2018}
P.~Gawron, D.~Kurzyk, and {\L}.~Pawela, ``{QuantumInformation}.jl{\textemdash}a
  julia package for numerical computation in quantum information theory,'' {\em
  {PLoS} {ONE}}, vol.~13, p.~e0209358, dec 2018.

\bibitem{heinosaari2019random}
T.~Heinosaari, M.~A. Jivulescu, and I.~Nechita, ``Random positive operator
  valued measures,'' {\em arXiv preprint arXiv:1902.04751}, 2019.

\end{thebibliography}
\end{document}